\documentclass[]{stripped}

\pdfoutput=1

\usepackage{hyperref}
\usepackage{amsfonts,amsmath,amsthm}
\usepackage{euscript}

\newcommand{\scr}{\EuScript}
\DeclareMathSymbol{\upset}{\mathopen}{symbols}{"22}
\DeclareMathSymbol{\downset}{\mathopen}{symbols}{"23}
\newcommand{\myF}{\scr F}
\newcommand{\myG}{\scr G}
\newcommand{\myW}{\scr W}
\newcommand{\supp}{\mathrm{supp}}
\newcommand{\upzeta}{\zeta^\uparrow}
\newcommand{\downzeta}{\zeta^\downarrow}

\newtheorem{Thm}{Theorem}
\newtheorem{Lem}[Thm]{Lemma}

\begin{document}


\title[]{The fast intersection transform\\with applications to counting paths}

\thanks{This research was supported in part by the Academy of Finland,
  Grants 117499 (P.K.) and 109101 (M.K.), and by the Swedish Research
  Council, project ``Exact Algorithms'' (A.B. and T.H.).}

\author[lu]{}{Andreas Bj\"orklund}
\address[lu]{Lund University, Department of Computer Science, P.O.Box 118, SE-22100 Lund, Sweden}

\author[itu]{}{Thore Husfeldt}
\address[itu]{IT University of Copenhagen, Rued Langgaards Vej 7, 2300
  K\o{}benhavn S, Denmark and Lund University, Department of Computer Science, P.O.Box 118, SE-22100 Lund, Sweden}

\author[he]{}{Petteri Kaski}

\author[he]{}{\mbox{Mikko Koivisto}}

\address[he]{Helsinki Institute for Information Technology HIIT,
  University of Helsinki, Department of Computer Science, P.O.Box 68,
  FI-00014 University of Helsinki, Finland\newline
    \ \newline
    \emph{E-mail addresses\/}:\ %
         {\tt andreas.bjorklund@logipard.com}, 
         {\tt thore.husfeldt@gmail.com}, \newline
    \phantom{\emph{E-mail addresses\/}:\ }%
         {\tt petteri.kaski@cs.helsinki.fi}, 
         {\tt mikko.koivisto@cs.helsinki.fi}}


\begin{abstract}
\noindent
We present an algorithm for evaluating a linear ``intersection transform'' 
of a function defined on the lattice of subsets of an $n$-element set.
In particular, the algorithm constructs an arithmetic circuit for 
evaluating the transform in ``down-closure time'' relative to the support 
of the function and the evaluation domain.
As an application, we develop an algorithm that,
given as input a digraph with $n$ vertices and bounded integer 
weights at the edges, counts paths by weight and given length 
$0\leq\ell\leq n-1$ in time $O^*(\exp(n\cdot H(\ell/(2n))))$, 
where $H(p)=-p\log p-(1-p)\log(1-p)$, and the notation $O^*(\cdot)$
suppresses a factor polynomial in $n$.
\end{abstract}



\maketitle


\section{Introduction}

Efficient algorithms for linear transformations,
such as the fast Fourier transform of Cooley and Tukey~\cite{CoTu65} 
and Yates' algorithm~\cite{Yate37},
are fundamental tools both in computing theory and in practical
applications. Therefore it is surprising that some arguably
elementary transformations have apparently not been investigated 
from an algorithmic perspective. 

This paper contributes by studying an ``intersection transform'' 
of functions defined on subsets of a ground set. 
In precise terms, let $U$ be a finite set with $n$ 
elements (the ground set), let $R$ be a ring, and 
denote by $2^U$ the set of all subsets of $U$.
The {\em intersection transform} maps a function
$f:2^U\rightarrow R$ to the function
$f\iota:\{0,1,\ldots,n\}\times 2^U\rightarrow R$, 
defined for all $j=0,1,\ldots,n$ and $Y\subseteq U$ by
\begin{equation}
\label{eq:intersection-transform}
f\iota_j(Y)=\sum_{\substack{X\subseteq U\\|X\cap Y|=j}} f(X).
\end{equation}

Our interest here is in particular to restrict (or ``trim'') 
the domains of the input $f$ and the output $f\iota$ from $2^U$ 
to given subsets of $2^U$.

For a subset $\myF\subseteq 2^U$, denote by $\downset\myF$ 
the {\em down-closure} of $\myF$, that is, the family of sets 
consisting of all the sets in $\myF$ and their subsets. 
The notation $O^*(\cdot)$ in what follows suppresses 
a factor polynomial in $n$.
The following theorem states our main result.

\begin{Thm}
\label{thm:fast-intersection-transform}
There exists an algorithm that, 
given $\myF\subseteq 2^U$ and $\myG\subseteq 2^U$ as input,
in time $O^*\bigl(|\downset\myF|+|\downset\myG|\bigr)$
constructs an $R$-arithmetic circuit with input gates
for $f:\myF\rightarrow R$ and output gates that
evaluate to $f\iota:\{0,1,\ldots,n\}\times\myG\rightarrow R$.
\end{Thm}

This result supplies yet another tool aimed 
at the resolution of a long-standing open problem, 
namely that of improving upon the classical (early 1960s) dynamic 
programming algorithm for the Travelling Salesman Problem (TSP). 
With an $O^*(2^n)$ running time for an instance with $n$ cities,
the classical algorithm, due to Bellman~\cite{Bell60,Bell62}, and, 
independently, Held and Karp~\cite{HeKa62}, remains the fastest known 
exact algorithm for the TSP. Moreover, progress has been
equally stuck at $O^*(2^n)$ even if one considers the more restricted
Hamiltonian Path (HP) and the Hamiltonian Cycle (HC) problems.

Armed with Theorem~\ref{thm:fast-intersection-transform}, 
we show that the $O^*(2^n)$ bound can be broken 
in a {\em counting} context, assuming one cares only for
{\em long} paths or cycles, as opposed to the {\em spanning} paths 
or cycles required by the TSP/HP/HC. 
(See \S\ref{sect:earlier} for a contrast with earlier work.)

Denote by $H$ the binary entropy function 
\begin{equation}
\label{eq:binent}
H(p)=-p\log p-(1-p)\log\,(1-p),\qquad 0\leq p\leq 1.
\end{equation}

\begin{Thm}
\label{thm:long-simple-walks}
There exists an algorithm that, given as input
\begin{itemize}
\item[(i)]
a directed graph $D$ with $n$ vertices and bounded integer 
weights at the edges,
\item[(ii)]
two vertices, $s$ and $t$, and
\item[(iii)]
a length $\ell=0,1,\ldots,n-1$,
\end{itemize}
counts, 
by total weight,
the number of paths of length $\ell$ from $s$ to $t$ in $D$ in time
\begin{equation}
\label{eq:long-simple-walks-runtime}
O^*\biggl(\exp\biggl(H\biggl(\frac{\ell}{2n}\biggr)\cdot n\biggr)\biggr).
\end{equation}
\end{Thm}

\medskip
For example, Theorem~\ref{thm:long-simple-walks} implies that 
we can count in $O(1.7548^n)$ time with length $\ell=0.5n$ 
and in $O(1.999999999^n)$ time with length $\ell=0.9999n$.
For length $\ell=n-1$ the bound reduces to the classical bound $O^*(2^n)$.

We observe that counting implies, by self-reducibility, 
that we can construct examples of the paths within the same time bound.
Similarly, we can count cycles of a given length within the same bound.
However, the efficient listing (in the form of vertex supports, 
weights, and ends $s,t$) of all the paths for 
any length $\ell\gg n/2$ appears not to be possible with present tools 
in $O\bigl((2-\epsilon)^n\bigr)$ time for $\epsilon>0$ independent
of $n$. Indeed, if it were possible, we would obtain the breakthrough 
$O\bigl((2-\epsilon)^n\bigr)$ algorithm for generic TSP by starting 
the classical algorithm from the output of the listing algorithm. 

We expect Theorem~\ref{thm:fast-intersection-transform} to have
applications beyond Theorem~\ref{thm:long-simple-walks};
for example, in the context of subset query problems discussed by 
Charikar, Indyk, and Panigrahy~\cite{ChIP02}.

Given $\myF\subseteq 2^U$ and $\myG\subseteq 2^U$
as input, we can count in $O^*\bigl(|\downset\myF|+|\downset\myG|\bigr)$ 
time for each $Y\in\myG$ the number of $X\in\myF$ that intersect
$Y$ in a given number of points; in particular, for each
$Y$ we can count the number of disjoint $X$. 

By duality of disjointness and set inclusion,  
we can thus count in $O^*\bigl(|\downset\myF|+|\upset\myG|\bigr)$ 
time for each $Y\in\myG$ the number of $X\in\myF$ with $X\subseteq Y$.
Here $\upset\myG$ denotes the {\em up-closure} of $\myG$, that is, 
the family of sets consisting of all the sets in $\myG$ and their 
supersets in $2^U$.

\subsection{Further remarks and earlier work}
\label{sect:earlier}

Theorem~\ref{thm:fast-intersection-transform} has 
its roots in Yates' algorithm~\cite{Yate37} for evaluating 
the product of a vector with the $n^{\mathrm{th}}$ 
Kronecker power of a $2\times 2$ matrix. 
While Yates' algorithm is essentially optimal, 
running in $O^*(2^n)$ ring operations given an input vector 
with $2^n$ entries, in certain cases the evaluation can be  
``trimmed'', assuming one requires only sporadic entries of the 
output vector. In particular, the present authors have observed 
\cite{BHKK08a} that the zeta and Moebius transforms on $2^U$ are 
amenable to trimming (see Lemma~\ref{lem:fzt} below for a 
precise statement). 

The proof of Theorem~\ref{thm:fast-intersection-transform} relies on
a trimmed concatenation of two ``dual'' zeta transforms, 
one that depends on supersets of a set (the ``up'' transform), 
and one that depends on subsets of a set (the ``down'' transform). 
To provide a rough intuition, we first use the up-zeta transform to 
drive information about $f$ on $\myF$ ``down'' to $\downset\myF$. 
Then we use a ``ranked''~\cite{BHKK07} down-zeta transform 
to assemble information ``up'' from $\downset\myG$ to $\myG$. 
Finally, we extract the intersection transform from the information 
gathered at each $Y\in\myG$. This essentially amounts to solving 
a fixed system of $R$-linear equations at each $Y\in\myG$.

This proof strategy yet again highlights a basic theme: 
the use of fast linear transformations to distribute and assemble 
information across a domain (e.g. time, frequency, subset lattice) 
so that ``local'' computations in the domain 
(e.g. pointwise multiplication, 
      solving local systems of linear equations) 
alternated with transforms enable the extraction of a desired result 
(e.g. convolution, intersection transform). 
Compared with earlier works such as \cite{BHKK07,BHKK08a,Kenn91},
the present approach establishes the serendipity of the up/down
dual transforms and introduces the ``linear equation trick'' into 
the toolbox of local computations.

Once Theorem~\ref{thm:fast-intersection-transform} is available,
Theorem~\ref{thm:long-simple-walks} stems from 
the observation that a path can be decomposed into 
two paths, each having half the length of the original path, 
with exactly one vertex in common. 
Theorem~\ref{thm:fast-intersection-transform} then enables
us to ``glue halves'' in $\myF$ and $\myG$,
where $\downset\myF$ and $\downset\myG$ consist of
sets of size at most $\lceil\ell/2\rceil+1$. 
This prompts the observation that 
Theorem~\ref{thm:fast-intersection-transform} is 
useful only when the bound $O^*\bigl(|\downset\myF|+|\downset\myG|\bigr)$
improves upon the trivial bound $O^*\bigl(|\myF||\myG|\bigr)$
obtained by a direct iteration over all pairs $(X,Y)\in\myF\times\myG$.

We know at least one alternative way of proving 
Theorem~\ref{thm:long-simple-walks}, without using
Theorem~\ref{thm:fast-intersection-transform}. Indeed,
assuming knowledge of trimming~\cite{BHKK08a}, one can use 
an algorithm of Kennes~\cite{Kenn91} to evaluate a sum
$\sum_{|Z|=j}\sum_{X\cap Y=Z} f(X)g(Y)$ for given 
$f:\myF\rightarrow R$ and $g:\myG\rightarrow R$
in $O^*\bigl(|\downset\myF|+|\downset\myG|\bigr)$ 
ring operations
(take the trimmed up-zeta transform of $f$ and $g$, 
take pointwise product of transforms, 
take the trimmed up-Moebius transform, and
sum over all $j$-subsets in $\downset\myF\cup\downset\myG$). 
This enables one to evaluate the right-hand side of \eqref{eq:wgf-len2} 
below in time \eqref{eq:long-simple-walks-runtime}, thus giving
an alternative proof of Theorem~\ref{thm:long-simple-walks}.

To contrast Kennes' algorithm with
Theorem~\ref{thm:fast-intersection-transform},
Kennes' algorithm computes for each $Z\subseteq U$ 
the sum over pairs $(X,Y)\in\myF\times\myG$ with $Z=X\cap Y$, 
whereas \eqref{eq:intersection-transform} computes, for each $Y\in\myG$ 
the sum over $X\in\myF$ with $|X\cap Y|=j$. Thus, Kennes' algorithm
provides control over the intersection $Z$ but lacks control
over the pairs $(X,Y)$, whereas \eqref{eq:intersection-transform} 
provides control over $Y$ but lacks control over the intersection 
(except for size).

As regards the TSP/HP/HC, earlier work on exact exponential-time 
algorithms can be divided roughly into three lines of study. 
(For a broader treatment of TSP/HP/HC and exact exponential-time 
algorithms, we refer to \cite{ABCC06,GuPu02,LLRS85}, and 
\cite{Woeg03}, respectively.)

One line of study has been to restrict the input graph, whereby
a natural restriction is to place an upper bound 
$\Delta$ on the degrees of the vertices.
Eppstein~\cite{Epps07} has developed an algorithm that runs
in time $O^*(2^{n/3})=O(1.260^n)$ for $\Delta=3$ and 
in time $O(1.890^n)$ for $\Delta=4$. 
Iwama and Nakashima \cite{IwNa07} have improved the $\Delta=3$ 
case to $O(1.251^n)$, and Gebauer~\cite{Geba08} the $\Delta=4$ case
to $O(1.733^n)$. The present authors established \cite{BHKK08b}
an $O\bigl((2-\epsilon)^n\bigr)$ bound for all $\Delta$,
with $\epsilon>0$ depending on $\Delta$ but not on $n$.

A second line of study has been to ease the space requirements of the
algorithms from exponential to polynomial in $n$. 
Karp~\cite{Karp82} and, independently, Kohn, Gottlieb, and 
Kohn~\cite{KoGK77} have shown that TSP with bounded integer weights
can be solved in time $O^*(2^n)$ and space polynomial in $n$.
Combined with restrictions on the graph, one can arrive
at running times $O^*\bigl((2-\epsilon)^n\bigr)$ and polynomial space
\cite{BHKK08b,Epps07,IwNa07}.

A third line of study relaxes the requirement on spanning paths/cycles
to ``long'' paths/cycles. In this setting, a simple backtrack algorithm 
finds a path of length $\ell$ in time $O^*(n^\ell)$.
Monien~\cite{Moni85} observed that this can be expedited
to $O^*(\ell!)$ time by a dynamic programming approach.
Alon, Yuster, and Zwick~\cite{AlYZ95} introduced 
a seminal colour-coding procedure and improved the
running time to $O^*((2e)^\ell)$ expected and 
$O^*(c^\ell)$ deterministic time, $c$ a large constant.
Subsequently, combining colour-coding ideas with a divide-and-conquer 
approach, Chen, Lu, Sze, and Zhang~\cite{CLSZ07}, and, independently,
Kneis, M\"olle, Richter, and Rossmanith~\cite{KMRR06},
developed algorithms with $O^*(4^\ell)$ expected and 
$O^*(4^{\ell+o(\ell)})$ deterministic time. 
A completely different approach was taken by Koutis~\cite{Kout08}, 
who presented an $O^*(2^{3\ell/2})$ expected time algorithm 
relying on a randomised technique for detecting whether a given 
$n$-variate polynomial, represented as an arithmetic circuit with only 
sum and product gates, has a square-free monomial of degree $\ell$ with 
an odd coefficient. Recently, Williams~\cite{Will08} extended
Koutis' technique and obtained an $O^*(2^\ell)$ expected time algorithm.

To contrast with Theorem~\ref{thm:long-simple-walks}, 
while the $O^*(2^\ell)$ bound of the Koutis--Williams~\cite{Kout08,Will08}
algorithm is superior to the bound \eqref{eq:long-simple-walks-runtime}
in Theorem~\ref{thm:long-simple-walks}, it is not immediate whether 
the Koutis--Williams approach extends to counting problems. 
Furthermore, it appears challenging to derandomise the Koutis--Williams
algorithm without increasing the running time (see~\cite[p.~6]{Will08}),
whereas the algorithm in Theorem~\ref{thm:long-simple-walks} is 
deterministic.

\section{The fast intersection transform}

\subsection{Preliminaries}
\label{sect:trimming-preliminaries}

For a logical proposition $P$, we use Iverson's bracket notation
$[P]$ to denote a $1$ if $P$ is true, and a $0$ if $P$ is false.

Let $\myF\subseteq 2^U$ and $f:\myF\rightarrow R$.

Define the {\em up-zeta transform} $f\upzeta$
for all $Y\subseteq U$ by
\begin{equation}
\label{eq:up-zeta}
f\upzeta(Y)=\sum_{\substack{X\in\myF\\Y\subseteq X}} f(X)\,.
\end{equation}

Define the {\em down-zeta transform} $f\downzeta$
for all $Y\subseteq U$ by
\begin{equation}
\label{eq:down-zeta}
f\downzeta(Y)=
\sum_{\substack{X\in\myF\\X\subseteq Y}} f(X)\,.
\end{equation}

The following lemma condenses the essential properties of 
the ``trimmed'' fast zeta transform \cite{BHKK08a}.

\begin{Lem} 
\label{lem:fzt}
There exist algorithms that construct, 
given $\myF\subseteq 2^U$ and $\myG\subseteq 2^U$ as input,
an $R$-arithmetic circuit with input gates for 
$f:\myF\rightarrow R$ and output gates that evaluate to
\begin{enumerate}
\item
$f\upzeta:\myG\rightarrow R$, 
with construction time $O^*\bigl(|\myF|+|\upset\myG|\bigr)$;
\item
$f\upzeta:\myG\rightarrow R$,
with construction time $O^*\bigl(|\downset\myF|+|\myG|\bigr)$;
\item
$f\downzeta:\myG\rightarrow R$,
with construction time $O^*\bigl(|\myF|+|\downset\myG|\bigr)$; and
\item
$f\downzeta:\myG\rightarrow R$,
with construction time $O^*\bigl(|\upset\myF|+|\myG|\bigr)$.
\end{enumerate}
\end{Lem}

\subsection{The inverse of truncated Pascal's triangle}

We work with the standard extension of the binomial
coefficients to arbitrary integers 
(see Graham, Knuth, and Patashnik~\cite{GrKP94}). 
For integers $p$ and $q$, we let
\begin{equation}
\label{eq:binom}
\binom{p}{q}=
\begin{cases}
\prod_{k=1}^q\frac{p+1-k}{k} & \text{if $q>0$;}\\
1                            & \text{if $q=0$;}\\
0                            & \text{if $q<0$.}
\end{cases}
\end{equation}

The following lemma is folklore, 
but we recall a proof here for convenience of exposition.

\begin{Lem}
\label{lem:ab}
The integer matrices $A$ and $B$ with entries
\begin{equation}
\label{eq:ab}
a_{ij}=\binom{j}{i},
\qquad
b_{ij}=(-1)^{i+j}\binom{j}{i},
\qquad i,j=0,1,\ldots,n
\end{equation}
are mutual inverses.
\end{Lem}
\begin{proof}
Let us first consider the $(i,j)$-entry of $AB$:
\[
\begin{split}
\sum_{k=0}^n a_{ik}b_{kj}
&=(-1)^j\sum_{k=0}^n\binom{k}{i}\binom{j}{k}(-1)^k\\
&=(-1)^j\sum_{k=i}^j\binom{k}{i}\binom{j}{k}(-1)^k\\
&=(-1)^{i+j}\binom{j}{i}\sum_{k=i}^j(-1)^{k-i}\binom{j-i}{k-i}\\
&=[i=j].
\end{split}
\]
Here the second equality follows by observing that $j\geq 0$
implies $\binom{j}{k}=0$ for all $k>j$; similarly, $k\geq 0$
implies $\binom{k}{i}=0$ for all $0\leq k<i$.
The third equality follows from an application of the 
identity $\binom{p}{q}\binom{q}{r}=\binom{p}{r}\binom{p-r}{q-r}$,
valid for all integers $p,q,r$ (see \cite[Equation 5.21]{GrKP94}).
The last equality follows from an application of the Binomial Theorem.

The analysis for the $(i,j)$-entry of $BA$ is similar:
\[
\begin{split}
\sum_{k=0}^n b_{ik}a_{kj}
&=(-1)^i\sum_{k=0}^n\binom{k}{i}\binom{j}{k}(-1)^k\\
&=(-1)^i\sum_{k=i}^j\binom{k}{i}\binom{j}{k}(-1)^k\\
&=(-1)^{i+i}\binom{j}{i}\sum_{k=i}^j(-1)^{k-i}\binom{j-i}{k-i}\\
&=[i=j].
\end{split}
\]
\end{proof}

It follows from Lemma~\ref{lem:ab} that the matrices $A$ and $B$ 
are mutual inverses over an arbitrary ring $R$, where the entries of 
the matrices are understood to be embedded into $R$ via the natural 
ring homomorphism $z\mapsto z_R=z\cdot 1_R$, where $1_R$ is the 
multiplicative identity element of $R$, and $z$ is an integer.

\subsection{Proof of Theorem~\ref{thm:fast-intersection-transform}}

We first describe the algorithm and then prove its correctness.
All arithmetic in the evaluations, and all derivations in 
subsequent proofs, are carried out in the ring $R$. 

Let $\myF\subseteq 2^U$ and $\myG\subseteq 2^U$ be given as input 
to the algorithm. The circuit is a sequence of three ``modules'' 
starting at the input gates for $f:\myF\rightarrow R$.

{\em $1$.~Up-transform.}
Evaluate the up-zeta transform
\begin{equation}
\label{eq:g-def}
g=f\upzeta\ \text{on $\downset\myF$}
\end{equation}
with a circuit of size $O^*\bigl(|\downset\myF|\bigr)$ using
Lemma~\ref{lem:fzt}(1). Observe that \eqref{eq:up-zeta}
implies that all nonzero values of $f\upzeta$ are in $\downset\myF$.

{\em $2$.~Down-transform by rank.}
For each $i=0,1,\ldots,n$, evaluate 
$g^{(i)}$, the component of $g$ with rank $i$, on $\downset\myF$;
that is, for all $X\in\downset\myF$, set
\begin{equation}
\label{eq:g-rank}
g^{(i)}(X)=\begin{cases}
g(X) & \text{if $|X|=i$;}\\
0    & \text{otherwise.}
\end{cases}
\end{equation}
Then, for each $i=0,1,\ldots,n$, evaluate 
\begin{equation}
\label{eq:yi-def}
y_i=g^{(i)}\downzeta\ \text{on $\myG$}
\end{equation}
with a circuit of size 
$O^*\bigl(|\downset\myF|+|\downset\myG|\bigr)$ using Lemma~\ref{lem:fzt}(3).

{\em $3$.~Recover the intersection transform.}
Let $B_R$ be the matrix in Lemma~\ref{lem:ab} with entries embedded to $R$.
Associate with each $Y\in\myG$ the column vector 
\[
y(Y)=\bigl(y_0(Y),y_1(Y),\ldots,y_n(Y)\bigr)^T.
\]
For each $Y\in\myG$, evaluate the column vector 
\[
x(Y)=\bigl(x_0(Y),x_1(Y),\ldots,x_n(Y)\bigr)^T
\]
as the matrix--vector product
\begin{equation}
\label{eq:x-def}
x(Y)=B_R\,y(Y).
\end{equation}
Because the matrix $B_R$ is fixed, this can be implemented with 
$O^*\bigl(|\myG|\bigr)$ fixed $R$-arithmetic gates.

The circuit thus consists of 
$O^*\bigl(|\downset\myF|+|\downset\myG|\bigr)$ 
$R$-arithmetic gates. It remains to show that the circuit
actually evaluates the intersection transform of $f$.

\begin{Lem}
For all $Y\in\myG$ and $j=0,1,\ldots,n$ it holds that $x_j(Y)=f\iota_j(Y)$.
\end{Lem}
\begin{proof}
Let $Y\in\myG$ and $i=0,1,\ldots,n$. 
Consider the following derivation:
\begin{equation}
\label{eq:ybf}
\begin{aligned}
y_i(Y)&=\sum_{\substack{Z\subseteq Y\\|Z|=i}}
          \sum_{\substack{X\in\myF\\Z\subseteq X}} f(X)
     \\
     &= \sum_{X\in\myF} f(X) 
          \sum_{\substack{Z\subseteq X\cap Y\\|Z|=i}} 1_R
     \\
     &= \sum_{X\in\myF}\, \binom{|X\cap Y|}{i}_{\!\!R}\, f(X) 
     \\
     &= \sum_{j=0}^n \,\,\binom{j}{i}_{\!\!R}\!
           \sum_{\substack{X\in\myF\\|X\cap Y|=j}} f(X) 
     \\
     &= \sum_{j=0}^n \,\bigl(a_{ij}\bigr)_{\!R} \,f\iota_j(Y).
\end{aligned}
\end{equation}
Here the first equality expands the definitions 
\eqref{eq:yi-def}, \eqref{eq:down-zeta}, \eqref{eq:g-rank}, 
\eqref{eq:g-def}, and \eqref{eq:up-zeta}.
The second equality follows by changing the order of summation
and observing that $Z\subseteq X\cap Y$ if and only if
both $Z\subseteq X$ and $Z\subseteq Y$.
The fourth equality follows by collecting the terms with
$|X\cap Y|=j$ together.
The last equality follows from \eqref{eq:ab} and 
\eqref{eq:intersection-transform}. 

Now let $j=0,1,\ldots,n$, and observe that 
\eqref{eq:x-def}, \eqref{eq:ybf}, and Lemma~\ref{lem:ab} imply
\[
\begin{split}
x_j(Y)
&=\sum_{i=0}^n \bigl(b_{ji}\bigr)_{\!R}\,y_i(Y)
\\
&=\sum_{i=0}^n \bigl(b_{ji}\bigr)_{\!R}
    \sum_{k=0}^n \bigl(a_{ik}\bigr)_{\!R}\, f\iota_k(Y)
\\
&=\sum_{k=0}^n 
    \biggl(\sum_{i=0}^n\, \bigl(b_{ji}\bigr)_{\!R}\, 
                          \bigl(a_{ik}\bigr)_{\!R}\biggr)
       f\iota_k(Y)
\\
&=\sum_{k=0}^n\,
    \biggl(\sum_{i=0}^n b_{ji}a_{ik}\biggr)_{\!\!R}\,
      f\iota_k(Y)
\\
&=\sum_{k=0}^n \,\bigl[j=k\bigr]_R\,f\iota_k(Y)
\\
&=f\iota_j(Y).
\end{split}
\]
\end{proof}

\section{Counting paths}

\subsection{Preliminaries}

We require some preliminaries before proceeding with the 
proof of Theorem~\ref{thm:long-simple-walks}.
For basic graph-theoretic terminology we refer to West~\cite{West01}.

Let $D$ be an $n$-vertex digraph with vertex set $V$ and edge set $E$,
possibly with loops and parallel edges. (However, to avoid
further technicalities in the bound \eqref{eq:long-simple-walks-runtime},
we assume that the number of edges in $D$ is bounded from above by 
a polynomial in $n$.)
Associated with each edge $e\in E$ is a {\em weight} 
$w(e)\in \{0,1,\ldots\}$. 
For an edge $e\in E$, denote by $e^-$ (respectively, $e^+$)
the start vertex (respectively, the end vertex) of $e$.

It is convenient to work with the terminology of walks
instead of paths.
A {\em walk} of {\em length} $\ell$ in $D$ is a tuple
$W=(v_0,e_1,v_1,e_2,v_2,\ldots,v_{\ell-1},e_\ell,v_\ell)$
such that $v_0,v_1,\ldots,v_\ell\in V$,
$e_1,e_2,\ldots,e_\ell\in E$, and, for each $i=1,2,\ldots,\ell$,
it holds that $e_i^{-}=v_{i-1}$ and $e_i^{+}=v_i$.
The walk $W$ is said to be {\em from} $v_0$ {\em to} $v_\ell$. 

A walk is {\em simple} if $v_0,v_1,\ldots,v_\ell$ are distinct vertices.
The set of distinct vertices occurring in a walk is the \emph{support} 
of the walk. We denote the support of a walk $W$ by $\supp(W)$.
The {\em weight} of a walk $W$ is the sum of the weights of the edges 
in the walk; a walk with no edges has zero weight.
We write $w(W)$ for the weight of $W$.

For $s,t\in V$ and $S\subseteq V$ we denote by
$\myW_{s,t}(S)$ the set of all simple walks 
from $s$ to $t$ with support $S$. Observe that
$\myW_{s,t}(S)$ is empty unless both $s\in S$ and $t\in S$.

Let $z$ be a polynomial indeterminate, and 
define an associated polynomial generating function by
\begin{equation}
\label{eq:wgf-supp}
f_{s,t}(S)=\sum_{W\in\myW_{s,t}(S)}z^{w(W)}.
\end{equation}
Put otherwise, the coefficient of each monomial $z^w$ 
of $f_{s,t}(S)$ enumerates the simple walks from 
$s$ to $t$ with support $S$ and weight $w$.

For $k=0,1,\ldots,n$, denote by $\binom{V}{k}$ the
set of all $k$-subsets of $V$.

For $\ell=0,1,\ldots,n-1$, define 
a polynomial generating function by 
\begin{equation}
\label{eq:wgf-len}
g_{s,t}(\ell)=\sum_{J\in\binom{V}{\ell+1}}f_{s,t}(J).
\end{equation}
Put otherwise, the coefficient of each
monomial $z^w$ of $g_{s,t}(\ell)$ enumerates the
simple walks from $s$ to $t$ with length $\ell$ and weight $w$.

\subsection{Proof of Theorem~\ref{thm:long-simple-walks}}

Let $B\in\{0,1,\ldots\}$ be fixed.
Let $D$ be a digraph with $n$ vertices and
edge weights $w(e)\in\{0,1,\ldots,B\}$ for all $e\in E$.
Let $s,t\in V$. Let $\ell=0,1,\ldots,n-1$.

With the objective of eventually applying 
Theorem~\ref{thm:fast-intersection-transform},
let $U=V$ and let $R$ be the univariate polynomial ring over $z$ 
with integer coefficients.

To compute $g_{s,t}(\ell)$, proceed as follows. 
First observe that the generating polynomials
\eqref{eq:wgf-supp} can be computed by the following recursion
on subsets of $V$.
The singleton sets $\{s\}\subseteq V$, $s\in V$,
form the base case of the recursion:
\begin{equation}
\label{eq:wgf-recbase}
f_{s,s}\bigl(\{s\}\bigr)=1.
\end{equation}
The recursive step is defined for all 
$s,t\in V$ and $S\subseteq V$, $|S|\geq 2$, by
\begin{equation}
\label{eq:wgf-recstep}
f_{s,t}(S)=\sum_{a\in S\setminus\{t\}}
            f_{s,a}\bigl(S\setminus\{t\}\bigr)
            \biggl(\sum_{\substack{e\in E\\e^-=a\\e^+=t}}z^{w(e)}\biggr).
\end{equation}

Now, using \eqref{eq:wgf-recbase} and \eqref{eq:wgf-recstep}, 
evaluate
\begin{equation}
\label{eq:psa-def}
p_{s,a}=f_{s,a}\ \text{on $\tbinom{V}{\lfloor \ell/2\rfloor+1}$}
\end{equation}
for each $a\in V$.
Then, using \eqref{eq:wgf-recbase} and 
\eqref{eq:wgf-recstep} again,
evaluate 
\begin{equation}
\label{eq:qat-def}
q_{a,t}=f_{a,t}\ \text{on $\tbinom{V}{\lceil \ell/2\rceil+1}$}.
\end{equation}
Next, using the algorithm in Theorem~\ref{thm:fast-intersection-transform}
with $\myF=\tbinom{V}{\lceil \ell/2\rceil+1}$ and 
$\myG=\tbinom{V}{\lfloor \ell/2\rfloor+1}$, evaluate 
\begin{equation}
\label{eq:rat-def}
r_{a,t}=q_{a,t}\iota_1\ \text{on $\tbinom{V}{\lfloor \ell/2\rfloor+1}$}
\end{equation}
Finally, evaluate the right-hand side of
\begin{equation}
\label{eq:wgf-len2}
 g_{s,t}(\ell)=
  \sum_{a\in V}
      \sum_{S\in\binom{V}{\lfloor \ell/2\rfloor+1}} 
        p_{s,a}(S)r_{a,t}(S)
\end{equation}
by direct summation. 

The entire evaluation can thus be carried out with an $R$-arithmetic
circuit of size
\begin{equation}
\label{eq:sizetime}
O^*\bigl(\big|\downset\tbinom{V}{\lceil \ell/2\rceil+1}\big|+
         \big|\downset\tbinom{V}{\lfloor \ell/2\rfloor+1}\big|\bigr)
\end{equation}
that can be constructed in similar time.

To justify the equality in \eqref{eq:wgf-len2},
consider the following derivation:
\[
\begin{split}
  \sum_{a\in V}&
    \sum_{S\in\binom{V}{\lfloor \ell/2\rfloor+1}} 
      p_{s,a}(S)r_{a,t}(S)\\
&=\sum_{a\in V}
    \sum_{S\in\binom{V}{\lfloor \ell/2\rfloor+1}}
      f_{s,a}(S)
    \sum_{\substack{T\in\binom{V}{\lceil\ell/2\rceil+1}\\|S\cap T|=1}} 
      f_{a,t}(T)\\
&=\sum_{a\in V}
    \sum_{S\in\binom{V}{\lfloor \ell/2\rfloor+1}}
      \sum_{\substack{T\in\binom{V}{\lceil\ell/2\rceil+1}\\|S\cap T|=1}} 
        \sum_{W_{sa}\in\myW_{s,a}(S)}
          \sum_{W_{at}\in\myW_{a,t}(T)} 
            z^{w(W_{sa})+w(W_{at})}\\
&=\sum_{a\in V}
    \sum_{S\in\binom{V}{\lfloor \ell/2\rfloor+1}}
      \sum_{\substack{T\in\binom{V}{\lceil\ell/2\rceil+1}\\S\cap T=\{a\}}} 
        \sum_{W_{sa}\in\myW_{s,a}(S)}
          \sum_{W_{at}\in\myW_{a,t}(T)} 
            z^{w(W_{sa})+w(W_{at})}\\
&=\sum_{J\in\binom{V}{\ell+1}}
    \sum_{W\in\myW_{s,t}(J)} z^{w(W)}\\
&=g_{s,t}(\ell).
\end{split}
\]
Here the first two equalities expand 
\eqref{eq:psa-def}, \eqref{eq:rat-def}, \eqref{eq:intersection-transform}, 
\eqref{eq:qat-def}, and \eqref{eq:wgf-supp}.
The third equality follows by observing that 
$\myW_{s,a}(S)$ and 
$\myW_{a,t}(T)$ are both nonempty only if
$a\in S$ and $a\in T$. Thus, $|S\cap T|=1$ implies that 
only terms with $S\cap T=\{a\}$ appear in the sum.
The fourth equality is justified as follows. 
First observe that an arbitrary walk $W$ of length $\ell$ 
from $s$ to $t$ has the property that 
there exists a $J\in\binom{V}{\ell+1}$ with $\supp(W)=J$ 
if and only if the walk is simple.
Moreover, a simple walk $W$ of length $\ell$ from 
$s$ to $t$ has a bijective decomposition $W\mapsto (W_{sa},W_{at})$ 
into two simple subwalks,
$W_{sa}$ and $W_{at}$, with $\supp(W_{sa})\cap\supp(W_{at})=\{a\}$
for some $a\in V$.
Indeed, $W_{sa}$ is the length-$\lfloor\ell/2\rfloor$
prefix of $W$ from $s$ to some $a\in V$, 
and $W_{at}$ is the length-$\lceil\ell/2\rceil$
suffix of $W$ from $a$ to $t$.
Conversely, prepend $W_{sa}$ to $W_{at}$, 
deleting one occurrence of $a$ in the process, to get $W$.
The fifth equality follows from \eqref{eq:wgf-len} and \eqref{eq:wgf-supp}.

It remains to analyse the total running time of constructing
and evaluating the circuit in terms of $n$ and $\ell$.

Because $B$ is fixed, all the ring operations are carried out
on polynomials of degree at most $Bn=O(n)$. 
Moreover, denoting by $m$ the number of edges in $D$, 
the coefficients in the polynomials are integers bounded
in absolute value by 
$2^m2^{5n}$, where $2^m$ is an upper bound for 
the coefficients in \eqref{eq:wgf-supp} and \eqref{eq:wgf-len}, 
and $2^{5n}$ is an upper bound for the expansion in intermediate 
values in the transforms. (Both bounds are far from tight.)
Recalling that we assume that $m$ is bounded from above by a polynomial 
in $n$, we have that the coefficients can be represented 
using a number of bits that is bounded from above by a polynomial in $n$. 
It follows that each ring operation runs in time bounded from above by a 
polynomial in $n$. 

To conclude that the algorithm runs within the claimed upper bound
\eqref{eq:long-simple-walks-runtime}, combine \eqref{eq:sizetime}
with the observation that for every $0< p\leq 1/2$ it holds that
\begin{equation}
\label{eq:binsum-entropy}
\sum_{k=0}^{\lfloor np\rfloor}\binom{n}{k}\leq \exp\bigl(H(p)\cdot n\bigr)
\end{equation}
where $H$ is the binary entropy function \eqref{eq:binent}.
(For a proof of \eqref{eq:binsum-entropy}, 
see Jukna~\cite[p.~283]{Jukn01}.)


\end{document}